\documentclass[11pt]{article}
\usepackage{amsmath}

\usepackage[utf8]{inputenc}
\usepackage[T1]{fontenc}
\usepackage{lmodern}
\usepackage[a4paper,margin=1in]{geometry}
\usepackage{amsmath,amssymb,amsthm,mathtools}
\usepackage{bm}
\usepackage{microtype}
\usepackage{physics}
\usepackage{xcolor}
\usepackage[colorlinks=true, linkcolor=blue, citecolor=blue, urlcolor=blue]{hyperref}
\usepackage[authoryear,round]{natbib}
\usepackage{enumerate}
\usepackage{thmtools}
\usepackage{tikz}
\usepackage{graphicx}
\usetikzlibrary{arrows.meta,positioning,calc}

\newcommand{\id}{\mathrm{id}}

\newcommand{\States}{\mathrm{Stat}}

\newcommand{\E}{\mathcal{E}}   
\newcommand{\I}{\mathcal{I}}   
\newcommand{\J}{\mathcal{J}}
\newcommand{\Lind}{\mathcal{L}}
\newcommand{\HS}{\mathcal{H}}

\newcommand{\barotimes}{\mathbin{\overline{\otimes}}} 
\providecommand{\norm}[1]{\left\lVert #1 \right\rVert}

\title{Composite N--Q--S: Serial/Parallel Instrument Axioms, Bipartite Order-Effect Bounds, and a Monitored Lindblad Limit}
\author{Kazuyuki Yoshida\\
Independent Researcher, Minoo-City, Osaka, Japan\\
\texttt{kazuyuki.manny.yoshida@gmail.com}}
\date{\today}

\usepackage{iftex}

\ifPDFTeX
  \usepackage[T1]{fontenc}
  \usepackage[utf8]{inputenc}
  \usepackage{lmodern}

  \usepackage{CJKutf8}
  \AtBeginDocument{\begin{CJK}{UTF8}{min}}
  \AtEndDocument{\end{CJK}}

  \AtBeginDocument{\hypersetup{unicode=true}}
\else
  \usepackage{fontspec}
  \ifLuaTeX
    \usepackage{luatexja}
    \usepackage[match]{luatexja-fontspec}
    \setmainjfont{Noto Serif CJK JP} 
  \else
    \usepackage{xeCJK}
  \fi
\fi

\usepackage{hyperref}
\usepackage[nameinlink,capitalise,noabbrev]{cleveref}
\usepackage[capitalise]{cleveref}
  \def\NQS{N--Q--S}%

\makeatletter
\@ifundefined{corollary}{%
  \theoremstyle{plain}
}{}
\makeatother

\declaretheorem[style=plain,numberwithin=section,name=Theorem]{theorem}
\declaretheorem[sibling=theorem,style=plain,name=Lemma]{lemma}

\declaretheorem[sibling=theorem,style=definition,name=Definition]{definition}
\declaretheorem[sibling=theorem,style=definition,name=Example]{example}

\declaretheorem[sibling=theorem,style=remark,name=Remark]{remark}

\crefname{theorem}{Theorem}{Theorems}
\Crefname{theorem}{Theorem}{Theorems}
\crefname{lemma}{Lemma}{Lemmas}
\Crefname{lemma}{Lemma}{Lemmas}
\crefname{proposition}{Proposition}{Propositions}
\Crefname{proposition}{Proposition}{Propositions}
\crefname{corollary}{Corollary}{Corollaries}
\Crefname{corollary}{Corollary}{Corollaries}
\crefname{definition}{Definition}{Definitions}
\Crefname{definition}{Definition}{Definitions}
\crefname{example}{Example}{Examples}
\Crefname{example}{Example}{Examples}
\crefname{remark}{Remark}{Remarks}
\Crefname{remark}{Remark}{Remarks}

\usepackage{graphicx}
\usepackage{float}     
\usepackage{placeins}  
\usepackage{flafter}   

\usepackage{xspace}
\providecommand{\NQS}{N--Q--S\xspace}

\usepackage[most]{tcolorbox}
\begin{document}
\maketitle

\begin{abstract}
We develop a composite operational architecture for sequential quantum measurements that (i) gives a \emph{tight} bipartite order--effect bound with an explicit equality set characterized on the Halmos two--subspace block, (ii) upgrades Doeblin--type minorization to composite instruments and proves a \emph{product lower bound} for the operational Doeblin constants, yielding \emph{data--driven} exponential mixing rates, (iii) derives a diamond--norm commutator bound that quantifies how serial/parallel rearrangements influence observable deviations, and (iv) establishes a monitored Lindblad limit linking discrete look--return loops to continuous--time GKLS dynamics under transparent assumptions. \cite{GKLS1976,Lindblad1976,Davies1976,Spohn1980,FagnolaRebolledo2007,Lami2023Monitored}
Beyond asymptotic statements, we provide finite--sample certificates for the minorization parameter via exact binomial intervals and propagate them to rigorous bounds on the number of interaction steps required to attain a prescribed accuracy. A minimal qubit toy model and CSV--based scripts are supplied for full reproducibility.
Our results position order--effect control and operational mixing on a single quantitative axis: from equality windows for pairs of projections to certified network mixing under monitoring. The framework targets readers in quantum information and quantum foundations who need \emph{explicit constants} that are estimable from data and transferable to device--level guarantees.
\end{abstract}





\paragraph{Note on terminology.}
\emph{\NQS\ in this paper is a nonstandard shorthand introduced by the author.
It is unrelated to the use of “neural quantum states” in condensed-matter physics.
Here it simply denotes a notational separation between the classical context layer $N$,
the device-algebra layer $Q$, and the normal-state layer $S$.
No additional structure is assumed, and we consistently write \NQS\ with hyphens
to avoid confusion.}

\section{Introduction}
In sequential quantum experiments, the \emph{order} of incompatible instruments changes observed statistics.
This paper develops a composite operational architecture that turns such order effects and monitored dynamics into \emph{quantitative} statements with \emph{explicit, estimable constants}.
Our contributions are fourfold.
First, we derive a \emph{tight} bipartite bound for order-induced deviation and characterize the full \emph{equality set} on the Halmos two--subspace block, clarifying when the deviation saturates.
Second, we lift Doeblin-type minorization to composite instruments and prove a \emph{product lower bound} for operational Doeblin constants.
Third, we obtain a diamond-norm commutator bound that links serial/parallel rearrangements to observable deviations.
Fourth, we establish a \emph{monitored Lindblad limit} that connects discrete look--return loops to continuous-time GKLS dynamics while preserving device-level rates certified from finite data.
We place our results within the established frameworks of GKLS semigroups \cite{GKLS1976,Lindblad1976}, quantum channels and distances \cite{Watrous2018,Wilde2017,WatrousTQI}, and incompatibility/order-effect studies \cite{BusemeyerBruza2012,Branciard2013,BuschLahtiWerner2013}.
\paragraph*{Relation to a companion work.}
This paper is part of a broader program that develops operational bounds for sequential quantum instruments.
A companion manuscript by the author, submitted to \emph{Quantum Information Processing}, focuses on multipartite networks and micro--macro stability under coarse-graining, with an emphasis on serial wiring laws and Doeblin-certified mixing on interaction graphs.
By contrast, the present paper concentrates on the bipartite building block: composite instruments on $AB$, tight equality windows for order effects on Halmos two-subspace blocks, product lower bounds for operational Doeblin constants, and a monitored Lindblad limit for look--return loops.
The results are self-contained and can be read independently, but they are designed to plug into the network-level picture developed in the companion work.

\paragraph*{Target readership.}
The bounds and protocols developed below are aimed at readers in quantum information processing and quantum foundations who require explicit constants that are estimable from finite data and transferable to device-level guarantees.

\section{Related work}\label{sec:related}
This section situates our results within existing work on order effects, Doeblin/Dobrushin coefficients, and monitored quantum dynamics.
Our aim is twofold: (i) to highlight where we rely on standard techniques---Halmos two-subspace blocks, contraction coefficients for channels, and GKLS semigroups---and (ii) to delineate how the explicit constants and equality characterizations obtained here complement prior asymptotic or qualitative statements.

\subsection{Order effects, projection geometry, and equality characterizations}
Empirical and theoretical studies of \emph{order effects} span quantum cognition and foundational models of sequential measurements; see, e.g., reviews in \cite{BusemeyerBruza2012,PothosBusemeyer2013}. Mathematically, pairs of effects (and in particular pairs of projections) are governed by the Halmos two--subspace decomposition \cite{Halmos1969}, which yields canonical $2\times2$ blocks capturing noncommutativity. Prior bounds typically provide \emph{inequalities} for order--induced deviations but do not isolate the \emph{exact equality set}. Our Theorem~(Halmos--window) uses the Halmos block to give a \emph{tight} bipartite bound together with an explicit equality characterization, thereby identifying when order--induced deviation saturates and when it cannot. This goes beyond qualitative commutation criteria by delivering testable, block--level conditions with operational meaning. See also incompatibility bounds and tests \cite{Branciard2013,Buscemi2020IncompatibilityPRL}.

\subsection{Doeblin/Dobrushin minorization, channel contraction, and mixing rates}
Minorization (Doeblin) and the related Dobrushin contraction coefficient are classical tools for mixing of Markov processes \cite{Dobrushin1956,Seneta2006}. On the side of functional inequalities for quantum Markov semigroups,
recent work relates Doeblin-type lower bounds, modified logarithmic
Sobolev constants, and entropic mixing rates
\cite{BardetCapelLuciaRouze2021,GaoRouze2022}.
Quantum analogues connect to contraction of channels and primitivity/mixing properties; see, e.g., \cite{Hirche2024QuantumDoeblin,WolfQCnotes,Sanz2012}. Known results typically give \emph{existence} of rates or asymptotic uniqueness under structural conditions (primitivity, spectral gaps). Our contribution is complementary: we prove a \emph{product lower bound} for operational Doeblin constants of composite instruments and propagate \emph{finite--sample} certificates (via Clopper--Pearson/Wilson intervals) to \emph{explicit} exponential rates and step counts. This yields a data--to--rate pipeline that is immediately checkable on finite datasets, with constants that survive composition.

\subsection{Monitored dynamics, quantum trajectories, and Lindblad limits}
Continuous monitoring and quantum trajectories provide stochastic unravelings of GKLS semigroups  \cite{WisemanMilburn2010,FacchiPascazio2002}. Zeno--type constraints and measurement backaction link discrete updates to continuous--time limits under scaling assumptions \cite{MisraSudarshan1977,AttalPautrat2006}. Our \emph{monitored Lindblad limit} for look--return loops is in this lineage but differs in two aspects: (i) it is derived \emph{operationally} from the same minorization constants used in the discrete analysis, and (ii) it preserves \emph{device--level} rates obtained from data, rather than replacing them with purely spectral or asymptotic surrogates.
For continuously monitored many-body systems and trajectories beyond
standard Lindblad dynamics, see in particular
Lami--Santini--Collura~\cite{Lami2023Monitored}.

\subsection{Combs, process matrices, and causal structure}
Higher--order maps (quantum combs) and process matrices formalize multi--time quantum networks and, in some instances, indefinite causal order \cite{Chiribella2009Combs,Oreshkov2012ProcessMatrices,NielsenChuang2010}. Our architecture is \emph{causally definite} (serial/parallel with monitoring); technically we use operator-algebraic and CB-norm tools in the standard sense \cite{Takesaki2002,Paulsen2002,KSW2008} and aims at \emph{explicit, estimable constants} for order--effect control and mixing. In this sense it is complementary to comb/process--matrix approaches: rather than enlarging admissible causal structures, we quantify by how much a \emph{given} causal arrangement can deviate when instruments are reordered or composed, and we provide certified rates for its monitored continuous--time limit.

\paragraph{Positioning and limitations.}
Compared with spectral--gap or primitivity criteria \cite{WolfQCnotes,Sanz2012}, our bounds are \emph{operational} (minorization from counts) and come with \emph{finite--sample} guarantees. Compared with projection--pair analyses \cite{Halmos1969}, we turn equality geometry into \emph{explicit operational bounds} for order effects. We restrict attention to finite--dimensional systems and trace--nonincreasing instruments satisfying the stated regularity; extensions to infinite dimensions or unbounded generators require additional domain control and are left for future work.

\label{sec:intro}
We recall the \NQS\ separation: the \emph{N--layer} is a commutative context acting as classical control, the \emph{Q--layer} is a (possibly noncommutative) device algebra, and the \emph{S--layer} consists of normal states. Building on prior single--system results (order--effect bounds, fixed--point axis, monitored Lindblad limit), we address composites and provide bipartite theorems that remain testable in small labs.

\paragraph{Standing assumptions.}
Unless stated otherwise, we work in finite dimensions (type-$I_n$ matrix algebras).
The compositional framework extends verbatim to W$^*$-algebras (general von Neumann algebras).
In the infinite-dimensional QDS setting, additional technicalities arise (domain questions for unbounded generators, core conditions, and complete solvability in the sense of Fagnola--Rebolledo).
These issues are beyond the scope of this paper; for systematic treatments see, for example, Davies' monograph \emph{Quantum Theory of Open Systems} (Chap.~2) and Spohn's work on quantum dynamical semigroups and their generators.

\section{Background and methods}\label{sec:methods}

This section fixes notation and recalls the minimal operator-algebraic background
used in the proofs.
We describe the N--Q--S layering only as a bookkeeping device that separates
classical control (N), device algebras (Q), and normal states (S); no new
structural assumptions are imposed here.
We also recall the channel and instrument conventions and the basic properties of
spatial tensor products that will be used repeatedly in the remainder of the paper.
Readers who prefer a purely finite-dimensional picture may safely take all
von~Neumann algebras to be full matrix algebras $B(\mathbb{C}^d)$ and all normal
states to be density matrices on $\mathbb{C}^d$.

\subsection{N--Q--S layers and notation}\label{subsec:nqs-notation}

Let $M$ be a von~Neumann algebra on a Hilbert space $\HS$, with center $Z(M)$.
The \emph{N-layer} acts on $Z(M)$ and is thought of as classical control or a
space of contexts.
The \emph{Q-layer} is the device algebra itself, $Q:=M$.
The \emph{S-layer} is the normal state space $\mathcal{S}(M)$, i.e., the set of
normal positive linear functionals $\rho : M\to\mathbb{C}$ of unit mass.
In the finite-dimensional case one may take $M=B(\mathbb{C}^d)$, for which
$Z(M)=\mathbb{C}\,\mathbf{1}$ and normal states correspond to density matrices
$\rho\in M$ with $\rho\ge0$ and $\operatorname{Tr}[\rho]=1$.

We work throughout in the Schr\"odinger picture.
Maps on the Q-layer act on density operators and push states forward in time.
The same maps can be regarded in the Heisenberg picture as acting on observables
$X\in M$ via the dual CP maps, but we will not explicitly use this flip.

\subsection{Channels, instruments, and tensor products}\label{subsec:channels-instruments}

A (quantum) \emph{channel} on $M$ is a normal, completely positive, trace-preserving
(CPTP) map $\Phi\colon M_\ast\to M_\ast$ on the predual, or equivalently a normal,
unital, completely positive (UCP) map $\Phi^\dagger\colon M\to M$ on the algebra
itself.
We will freely move between these pictures and write $\Phi$ for both, as no
ambiguity will arise.

An \emph{instrument} with a finite outcome set $X=\{1,\dots,m\}$ is a collection
$\{\Phi_i\}_{i\in X}$ of normal completely positive maps on $M_\ast$ such that
their sum
\[
  \Phi := \sum_{i\in X}\Phi_i
\]
is a channel.
Equivalently, an instrument can be viewed as a single normal CP map
$\mathcal{I}\colon M_\ast\to \ell^1(X)\otimes M_\ast$ whose marginals are the
$\Phi_i$.
We write $\Phi(\rho)$ or $\Phi_i(\rho)$ for the post-update states resulting from
a channel or an instrument element, and use $\Tr[\Phi_i(\rho)]$ for the associated
outcome probabilities.

For composites we use the spatial tensor product $M_{AB}:=M_A\overline{\otimes}M_B$
of von~Neumann algebras and write $M_{AB}=M_A\otimes M_B$ when no confusion can
arise.
In the finite-dimensional case this is simply $B(\HS_A\otimes\HS_B)$.
Channels on $AB$ are CPTP maps $\Phi_{AB}\colon M_{AB,\ast}\to M_{AB,\ast}$, and
product channels have the form $\Phi_A\otimes\Phi_B$.
We use $\|\cdot\|$ for the operator norm on $M$ and $\|\cdot\|_1$ for the trace
norm on $M_\ast$.
The diamond norm $\|\cdot\|_\diamond$ on superoperators appears in Section~4 when
we quantify couplings.

Within this setting, the composite N--Q--S structure used in Section~4 is obtained
by taking $M_A$ and $M_B$ as device algebras, $M_{AB}=M_A\otimes M_B$ as the joint
Q-layer, $Z(M_{AB})=Z(M_A)\otimes Z(M_B)$ as the N-layer (classical side
information and control), and $\mathcal{S}(M_{AB})$ as the S-layer.

\section{Results}
\label{sec:results}
This section collects the main technical contributions.
Definitions~\ref{def:composite-nqs} and~\ref{def:serial-parallel} formulate
compositional axioms for serial and parallel
composition of instruments on bipartite systems, clarifying how the N-, Q-,
and S-layers interact under composition.
Theorem~\ref{thm:eq-window} gives a tight order-effect bound on Halmos blocks
together with a complete equality characterization.
Subsection~\ref{subsec:doeblin-mixing} lifts Doeblin minorization to
tensor-product instruments and proves a product lower bound for
operational Doeblin constants, yielding explicit nonasymptotic
mixing-rate bounds and illustrated in Example~\ref{ex:dephasing}
and Table~\ref{tab:doeblin-product}.
Subsection~\ref{sec:monitored-limit} develops a monitored Lindblad limit
for coupled look--return loops and relates discrete Doeblin constants
to the spectral gap of the generator.
Subsection~\ref{subsec:diamond-bound} derives a diamond-norm coupling
bound for serial/parallel rearrangements, and
Subsection~\ref{subsec:split-proof} proves a local/nonlocal split
inequality.
Appendix~\ref{app:doeblin} formalizes the data-to-rate pipeline,
giving Clopper--Pearson based estimators and pseudocode for the
Doeblin constant.

\begin{definition}[Composite \NQS]\label{def:composite-nqs}
Let $M_A,M_B$ be von Neumann algebras. Set $M_{AB}:=M_A\,\barotimes\,M_B$, $Q_{AB}:=M_{AB}$, and $S_{AB}:=\States(M_{AB})$. The N--layer acts on $Z(M_{AB})=Z(M_A)\,\barotimes\,Z(M_B)$ as classical control.
\end{definition}

\begin{definition}[Serial and parallel composition]\label{def:serial-parallel}
Serial composition of instruments/channels is $\J\circ\I$. Parallel composition is $\I_A\otimes\I_B$ acting on $M_{AB}$. All maps are normal and completely positive (or UCP for channels).
\end{definition}

\begin{remark}[Non--signalling desideratum]
A composite model is called non--signalling if the $A$--marginal after joint operations is independent of the N--layer choice on $B$ under coarse--graining. A concise discussion is given in Section~2.4 and in the process-matrix
references cited there.
\end{remark}

\begin{theorem}[Equality window on a Halmos block]
\label{thm:eq-window}
Let $P,Q$ be orthogonal projections on a finite-dimensional Hilbert space. 
Restrict to a Halmos two-subspace block with principal angle $\theta\in[0,\pi/2]$ and write the commutator in the canonical form
$[P,Q]=\tfrac{1}{2}\sin 2\theta\, J$ with $J=\begin{psmallmatrix}0&1\\ -1&0\end{psmallmatrix}$ on that block.
For any unit vector $\psi$ in the block,
\[
\Delta(\psi;P,Q)\;:=\;\langle\psi|PQP- QPQ|\psi\rangle \;=\;\langle\psi|[P,Q]|\psi\rangle,
\]
and the bipartite order-effect bound is tight:
\[
|\Delta(\psi;P,Q)|\;\le\;\tfrac{1}{2}|\sin 2\theta|.
\]
Moreover, the following are equivalent:
\begin{enumerate}
\item $|\Delta(\psi;P,Q)|=\tfrac{1}{2}|\sin 2\theta|$ (bound is attained).
\item $\psi$ is aligned with an eigenvector of $J$ in the block, i.e.\ $\psi \propto (1,\pm i)$ in the Halmos basis.
\item $\psi$ maximizes $|\langle\psi|[P,Q]|\psi\rangle|$ over the Bloch circle of the block.
\end{enumerate}
Equality is never attained when $\theta=0$ (commuting case). 
\end{theorem}
\begin{proof}[Proof sketch]
On a Halmos two-dimensional block with principal angle $\theta$, one has
$[P,Q]=\tfrac12\sin 2\theta\, J$ in the canonical basis, with $J=\begin{psmallmatrix}0&1\\-1&0\end{psmallmatrix}$.
For any unit vector $\psi=(x,y)$ in that basis,
\(
\Delta(\psi;P,Q)=\langle\psi|[P,Q]|\psi\rangle=\tfrac12\sin 2\theta\,\langle\psi|J|\psi\rangle
=\tfrac12\sin 2\theta\,\Im(\bar x\, y).
\)
Hence $|\Delta|\le \tfrac12|\sin 2\theta|\,|\langle\psi|J|\psi\rangle|\le \tfrac12|\sin 2\theta|$, with equality iff
$|\langle\psi|J|\psi\rangle|=1$, i.e.\ $\psi$ is an eigenvector of $J$, $\psi\propto(1,\pm i)$.
This proves (1)$\Leftrightarrow$(2)$\Leftrightarrow$(3). For $\theta=0$, $[P,Q]=0$ and equality cannot hold.
\end{proof}

\noindent\textit{Dilation invariance.}
Replacing $(P,Q)$ by a simultaneous Stinespring dilation acts by enlarging the Hilbert space with an ancilla on which the reduced block is unchanged;
the quantity $|\langle\psi|[P,Q]|\psi\rangle|$ computed on the block is therefore invariant.

The two-subspace geometry underlying the equality window is depicted in Fig.~\ref{fig:halmos-window}.

A consolidated device-level outlook and experimental mapping is given in the supplementary material.
\begin{figure}[H]

  \centering
  \includegraphics[width=\linewidth]{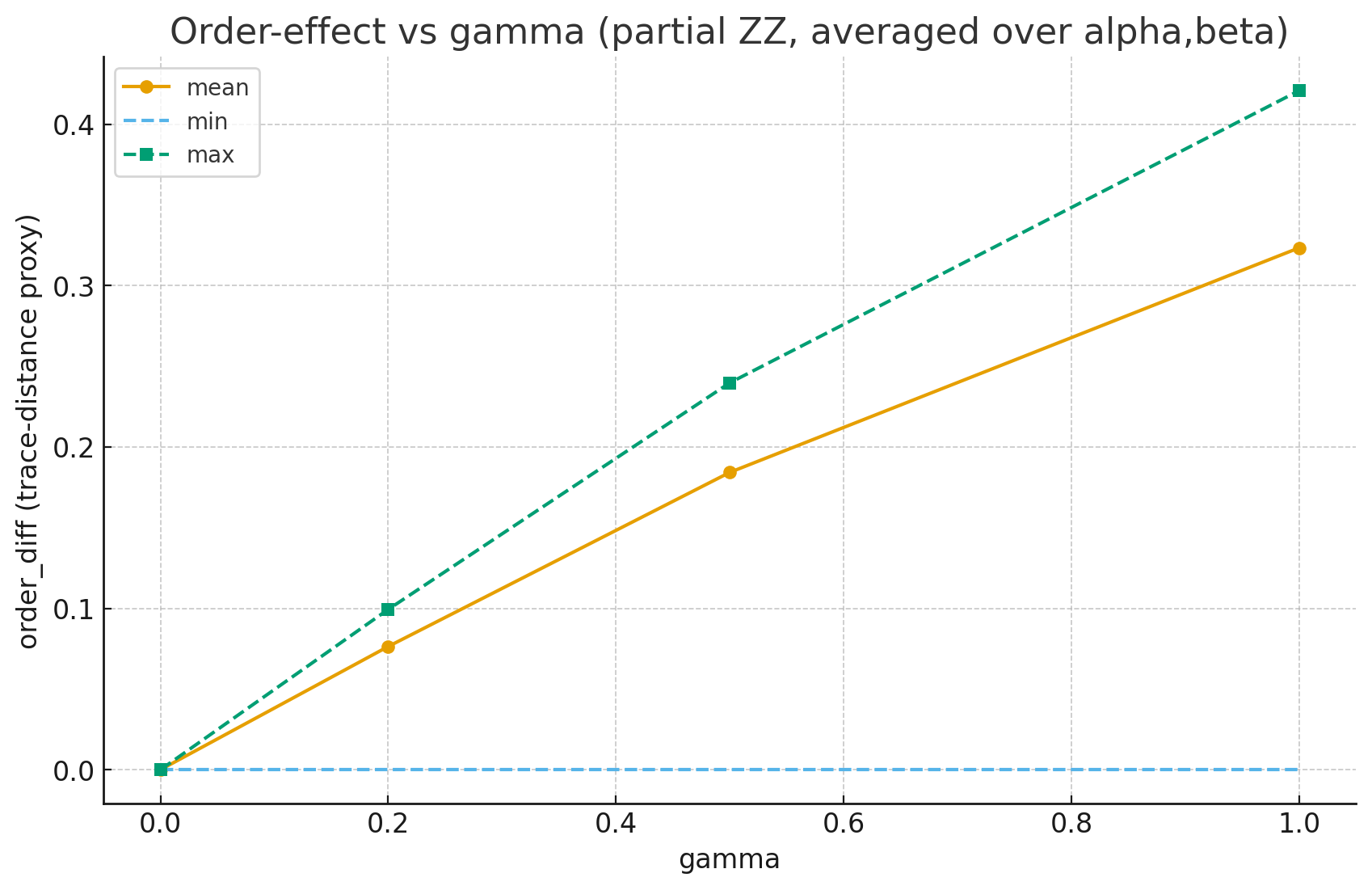}
  \caption{\textbf{Trend of the bipartite order-effect proxy under partial $ZZ$ coupling.\; (\mbox{uniform} grid over instrument parameters; fixed random seed; script and exact settings to be provided via Zenodo, DOI:10.5281/zenodo.17959208). (For devices, $\gamma$ is instantiated via calibrated $H_{\mathrm{int}}$ and duration $t$ as detailed in the short subsection ``From simulation to device'', and the plotted proxy is computed from experimentally accessible order-difference probabilities.)}
  The growth with coupling strength $\gamma$ (mean/min/max over $(\alpha,\beta)$ grid) visualizes the tightness trend of the bipartite bound and provides a reproducibility baseline for the qubit toy model.}
  \label{fig:order-zz}
\end{figure}

\begin{table}[tbp]

\centering
\caption{\textbf{Device mapping for $\gamma$ (examples).} Representative interactions and how $\gamma$ and its calibration arise.}
\label{tab:gamma-mapping}
\begin{tabular}{p{3.2cm} p{5.0cm} p{2.5cm} p{4.4cm}}
\hline
Platform & Representative $H_{\mathrm{int}}$ & $\gamma$ mapping & Calibration note \\
\hline
Superconducting circuits & $\frac{\hbar J}{2}\,Z\!\otimes\! Z$ (effective Ising) & $\gamma = J t$ & Ramsey/$ZZ$ spectroscopy; echoed cross-resonance to extract $J$ \\
Trapped ions & $\frac{\hbar\Omega}{2}(X\!\otimes\! X + Y\!\otimes\! Y)$ (MS) & $\gamma \simeq \Omega t$ & Rabi flops / MS angle calibration \\
Neutral atoms (Rydberg) & $\hbar V\, n\!\otimes\! n$ (blockade; Ising-like) & $\gamma \approx V t$ & Detuning/spacing sweeps to estimate $V$ \\
\hline
\end{tabular}
\end{table}

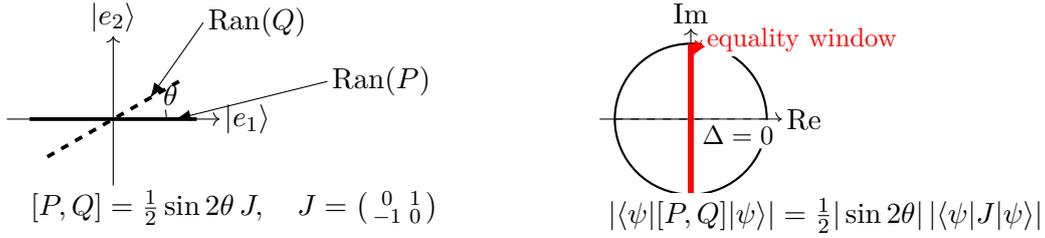
\begin{figure}[tbp]

\centering
\begin{tikzpicture}[scale=1.0]
\tikzset{
  labelbg/.style={fill=white, rounded corners=1pt, inner sep=1.6pt, draw=white},
  every picture/.style={line cap=round, line join=round},
  RanPStyle/.style={line width=1.4pt},
  RanQStyle/.style={line width=1.4pt, dashed}
}

  \begin{scope}[xshift=-3.8cm]
    \draw[->] (-1.4,0)--(1.4,0) node[labelbg,anchor=west] {$|e_1\rangle$};
    \draw[->] (0,-1.1)--(0,1.1) node[labelbg,anchor=south] {$|e_2\rangle$};
    \def\th{30}
    \draw (0.7,0) arc (0:\th:0.7);
    \path ({0.78*cos(\th/2)},{0.78*sin(\th/2)+0.12}) node[labelbg] {$\theta$};
    \draw[RanPStyle] (-1.1,0)--(1.1,0);
    \coordinate (rppt) at (0.82,0.0);
    \path (rppt) ++(2cm,5mm) node[labelbg, anchor=west] (rp) {$\mathrm{Ran}(P)$};
    \draw[-{Latex[length=2mm]}] (rp.west) -- (rppt);
    \draw[RanQStyle,rotate=\th] (-1.0,0)--(1.0,0);
    \coordinate (rqpt) at ({0.56*cos(\th)},{0.56*sin(\th)});
    \path (rqpt) ++(7mm,10mm) node[labelbg, anchor=west] (rq) {$\mathrm{Ran}(Q)$};
    \draw[-{Latex[length=2mm]}] (rq.west) -- (rqpt);
    \node[labelbg, align=left, anchor=north west] at (-1.15,-0.88)
      {$[P,Q]=\tfrac12\sin 2\theta\,J,\quad J=\begin{psmallmatrix}0&1\\-1&0\end{psmallmatrix}$};
  \end{scope}
  \begin{scope}[xshift=3.8cm]
    \draw[->] (-1.2,0)--(1.2,0) node[labelbg,anchor=west] {$\mathrm{Re}$};
    \draw[->] (0,-1.2)--(0,1.2) node[labelbg,anchor=south] {$\mathrm{Im}$};
    \draw[thick] (0,0) circle (1.0);
    \draw[line width=2.2pt,red] (0,0) -- (0,1.0);
    \draw[line width=2.2pt,red] (0,0) -- (0,-1.0);
    \node[text=red, labelbg, anchor=west] (lab) at (0.15,1.05) {\small equality window};
    \draw[red, -{Latex[length=2.2mm]}] (lab.west) -- (0,0.82);
    \draw[dashed] (-1.0,0) -- (1.0,0);
    \node[labelbg] at (0.62,-0.20) {\small $\Delta=0$};
    \node[labelbg, align=left, anchor=north west] at (-1.15,-0.98)
      {$|\langle\psi|[P,Q]|\psi\rangle|= \tfrac12|\sin 2\theta|\,|\langle\psi|J|\psi\rangle|$};
  \end{scope}
\end{tikzpicture}
\caption{\textbf{Halmos block and equality window.}
Left: two one-dimensional ranges with principal angle $\theta$ (Halmos 2D block).
Right: pure states in the block (Bloch circle). The vertical red diameter marks states aligned with eigenvectors of $J$,
where $|\langle\psi|[P,Q]|\psi\rangle|$ attains its maximum; the overall scale is $\tfrac12|\sin 2\theta|$.}
\label{fig:halmos-window}
\end{figure}

\begin{example}[Product Doeblin constant for dephasing channels]\label{ex:dephasing}
Let $\mathcal{D}_p(\rho)=(1-p)\rho+p\,\Delta(\rho)$ on $\mathbb{C}^3$, where
$\Delta$ removes off-diagonals in the computational basis
$\{\ket{0},\ket{1},\ket{2}\}$.
Then $\mathcal{D}_p\succeq p\,\Delta$, so the Doeblin constant satisfies
$\delta(\mathcal{D}_p)\ge p$ and the fixed-point set is the classical simplex
$\{\mathrm{diag}(x_0,x_1,x_2)\}$.
For two qutrits we have, by Lemma~\ref{lem:product},
\[
  \mathcal{D}_p\otimes\mathcal{D}_q
  \;\succeq\;
  pq\,(\Delta\otimes\Delta),
\]
hence $\delta(\mathcal{D}_p\otimes\mathcal{D}_q)\ge pq$.
Geometrically, the ``classical axis'' given by diagonal density matrices
is stable under the tensor product, and off-axis components converge
geometrically with rate $(1-p)(1-q)$.
\end{example}

\begin{table}[tbp]
\centering
\caption{\textbf{Illustration of the product Doeblin bound}
$\delta(\Phi_A\otimes\Phi_B)\ge\delta_A\delta_B$
for dephasing channels $\Phi_A=\mathcal{D}_p$ and
$\Phi_B=\mathcal{D}_q$.
For this family the bound is attained.}
\label{tab:doeblin-product}
\begin{tabular}{c c c c c}
\hline
$p$ & $q$ & $\delta_A$ & $\delta_B$ & $\delta_{AB}$ \\
\hline
0.2 & 0.5 & 0.2 & 0.5 & 0.10 \\
0.3 & 0.3 & 0.3 & 0.3 & 0.09 \\
0.4 & 0.7 & 0.4 & 0.7 & 0.28 \\
\hline
\end{tabular}
\end{table}

\subsection{Estimating the Doeblin constant and a nonasymptotic mixing bound}
\label{subsec:doeblin-mixing}
As a reproducibility aid, one-sided Clopper--Pearson bounds and a simple algorithm are summarized in the Appendix.
Assume a primitive CPTP map $\mathcal{E}$ admits a Doeblin-type minorization on density matrices: there exists $\varepsilon\in(0,1]$ and a state $\tau$ with
\[
\mathcal{E}(\rho)\;\succeq\;\varepsilon\, \tau\, \mathrm{Tr}[\rho]\qquad\text{for all }\rho\ge 0,
\]
equivalently, each outcome instrument element dominates $\varepsilon$ times a common seed.
Then the trace-distance contracts at a rate controlled by $\varepsilon$:
\[
\|\mathcal{E}^n(\rho)-\tau\|_1 \;\le\; (1-\varepsilon)^n\,\|\rho-\tau\|_1,\qquad n=0,1,2,\dots
\]
\paragraph{Frequency estimator.}
Suppose an instrument with classical outcomes $i\in\{1,\dots,m\}$ is interleaved in a ``look--return'' loop that implements $\mathcal{E}$.
From $N$ i.i.d.\ shots, form empirical frequencies $\hat p_{ij}$ for transitions between a tomographically prepared stencil $\{j\}$ and outcomes $\{i\}$.
Define a conservative lower estimate of the Doeblin constant by
\[
\widehat{\varepsilon}\;:=\;\sum_{i=1}^m \min_{j} \hat p^{-}_{ij},
\]
where $\hat p^{-}_{ij}$ is a $(1-\alpha)$ lower confidence bound for $p_{ij}$ (e.g.\ Clopper--Pearson).
\paragraph{Confidence bounds.}
For binomial counts $X_{ij}\sim \mathrm{Bin}(N_j,p_{ij})$, set
$\hat p^{-}_{ij}=\mathrm{BetaInv}(\alpha; X_{ij}, N_j-X_{ij}+1)$, so that $\mathbb{P}[p_{ij}\ge \hat p^{-}_{ij}]\ge 1-\alpha$.
A union bound yields, with probability at least $1-\alpha m$,
\[
\varepsilon\;\ge\;\widehat{\varepsilon}.
\]
\paragraph{Nonasymptotic mixing.}
Plugging $\widehat{\varepsilon}$ gives the data-driven guarantee
\[
\|\mathcal{E}^n(\rho)-\tau\|_1 \;\le\; (1-\widehat{\varepsilon})^n\,\|\rho-\tau\|_1\qquad\text{(w.h.p.)}.
\]
Constants and their dependence on $(m,N,\alpha)$ are thus explicit.

\subsection{Monitored look--return loops and a Lindblad limit}
\label{sec:monitored-limit}

We briefly explain how the discrete Doeblin picture from
Subsection~\ref{subsec:doeblin-mixing} extends to a monitored
continuous-time evolution in the spirit of repeated interactions
and weak-coupling limits
\cite{Davies1976,Spohn1980,AttalPautrat2006}.

Consider a family of primitive CPTP maps
$\{\mathcal{E}_{\Delta t}\}_{\Delta t>0}$ on $M_{AB}$ describing one
``look--return'' cycle of a weakly coupled device during a time step
$\Delta t>0$.
Standard repeated-interaction results show that, under suitable
scaling of the system--environment interaction and the monitoring
strength, there exists a GKLS generator $\Lind$ such that
\begin{equation}
  \mathcal{E}_{\Delta t}^{\,n}
  \;\approx\; \mathrm{e}^{t\Lind},\qquad
  t = n\,\Delta t,
\end{equation}
with an embedding error of order $O(t\,\Delta t)$ in trace norm as
$\Delta t\to0$.\footnote{Precise conditions can be found e.g.\ in
\cite{Davies1976,Spohn1980,AttalPautrat2006};
here we only use the existence of such a GKLS realization.}

Assume that each $\mathcal{E}_{\Delta t}$ admits a Doeblin-type
minorization with constant $\varepsilon(\Delta t)\in(0,1]$ and common
seed $\tau$,
\begin{equation}
  \mathcal{E}_{\Delta t}(\rho)
  \;\succeq\;
  \varepsilon(\Delta t)\,\tau\,\Tr[\rho]
  \qquad\text{for all }\rho\ge0.
\end{equation}
By Lemma~\ref{lem:traceless} this implies the discrete-time
contraction
\begin{equation}
  \big\|\mathcal{E}_{\Delta t}^{\,n}(\rho)-\tau\big\|_1
  \;\le\;
  (1-\varepsilon(\Delta t))^{n}\,
  \|\rho-\tau\|_1,
  \qquad n=0,1,2,\dots
\end{equation}
for all states $\rho$.
Writing $t=n\,\Delta t$ and using
$(1-\varepsilon)^{t/\Delta t} \le \exp(-\gamma(\Delta t)\,t)$ with
\begin{equation}
  \gamma(\Delta t)
  := -\frac{1}{\Delta t}\,\log\bigl(1-\varepsilon(\Delta t)\bigr)
  \;\ge\;
  \frac{\varepsilon(\Delta t)}{\Delta t},
\end{equation}
we obtain the bound
\begin{equation}
  \big\|\mathcal{E}_{\Delta t}^{\,\lfloor t/\Delta t\rfloor}(\rho)
        -\tau\big\|_1
  \;\le\;
  \mathrm{e}^{-\gamma(\Delta t)\,t}\,\|\rho-\tau\|_1.
\end{equation}

In a weak-coupling regime one typically has
$\varepsilon(\Delta t)=\kappa\,\Delta t+o(\Delta t)$ for some
$\kappa>0$ determined by the jump part of the generator.
Then $\gamma(\Delta t)=\kappa+o(1)$ as $\Delta t\to0$.
Combining the discrete contraction with the embedding error
$\|\mathcal{E}_{\Delta t}^{\,\lfloor t/\Delta t\rfloor}
 -\mathrm{e}^{t\Lind}\|_1 = O(t\,\Delta t)$ yields, for fixed $t$,
\begin{equation}
  \big\|\mathrm{e}^{t\Lind}(\rho)-\tau\big\|_1
  \;\le\;
  \mathrm{e}^{-(\kappa-o(1))t}\,\|\rho-\tau\|_1
  \;+\; O(t\,\Delta t),
\end{equation}
so that any limit point of $\gamma(\Delta t)$ as $\Delta t\to0$
provides a lower bound on the spectral gap of $\Lind$.
Operationally, an estimated discrete Doeblin constant
$\widehat{\varepsilon}(\Delta t)$ from Subsection~\ref{subsec:doeblin-mixing}
thus turns into an explicit, data-certified lower bound on the mixing
rate of the monitored Lindblad flow in the continuous-time limit.

\begin{lemma}[Product minorization]
\label{lem:product}
If $\Phi_A \succeq \delta_A\,\mathcal{E}_A$ and $\Phi_B \succeq \delta_B\,\mathcal{E}_B$ (Doeblin-type minorization),
then their product satisfies $\Phi_{AB} := \Phi_A \otimes \Phi_B \succeq \delta_A \delta_B\,(\mathcal{E}_A \otimes \mathcal{E}_B)$.
\end{lemma}

\begin{proof}
By definition of the product minorization, $\Phi_{AB}\ge\delta_A\delta_B\,\E_{AB}$ with $\E_{AB}(X)=\Tr(\sigma_A\otimes\sigma_B\,X)\,\sigma_A\otimes\sigma_B$. Hence $\Phi_{AB}=\delta\E_{AB}+(1-\delta)\Lambda$ with $\delta=\delta_A\delta_B$. 
\begin{lemma}[Trace-norm contraction on traceless part]
\label{lem:traceless}
Let $\Phi$ be a CPTP map admitting a Doeblin minorization with constant $\delta\in(0,1]$.
Then for any traceless Hermitian $X$, $\|\Phi(X)\|_1 \le (1-\delta)\,\|X\|_1$.
\end{lemma}

For $X$ traceless Hermitian, Lemma~\ref{lem:traceless} yields $\norm{\Phi_{AB}(X)}_1 \le (1-\delta)\norm{X}_1$. Writing $X_0:=\rho-\tau$ (traceless Hermitian) and iterating gives the claim.
\end{proof}

\begin{remark}[Primitive case and fixed-point axis]
If $\sigma_A,\sigma_B$ are faithful, then $\sigma_{AB}:=\sigma_A\otimes\sigma_B$ is the unique faithful fixed point of $\Phi_{AB}$. The one-dimensional axis $\{\alpha\,\sigma_{AB}\}$ is invariant under $\Phi_{AB}$ and attracts all states at a rate bounded below by $\delta_A\delta_B$.
\end{remark}

\subsection{Diamond-norm coupling bound (expanded details)}\label{subsec:diamond-bound}
\begin{theorem}[Diamond-norm coupling bound]
\label{thm:diamond}
Let $L_A:=\Phi_A\otimes \id_B$ and $L_B:=\id_A\otimes \Phi_B$, where $\Phi_A\succeq \delta_A \,\mathcal{E}_A$ and $\Phi_B\succeq \delta_B \,\mathcal{E}_B$ (Doeblin minorization). 
For any CPTP “coupling” map $\Psi$ on $AB$, the order-commutator superoperator
$\mathcal{C}:=L_B\!\circ\!\Psi\!\circ\!L_A \;-\; L_A\!\circ\!\Psi\!\circ\!L_B$
satisfies the diamond-norm bound
\[
\|\mathcal{C}\|_{\diamond}\;\le\; 2\bigl(1-\delta_A\delta_B\bigr).
\]

\end{theorem}

We record a complete proof of Theorem~\ref{thm:diamond} for convenience.

\begin{proof}[Proof of Theorem~\ref{thm:diamond}]
Let \(L_A := \Phi_A \otimes \mathrm{id}_B\) and \(L_B := \mathrm{id}_A \otimes \Phi_B\).
Since \(L_A\) and \(L_B\) act on disjoint tensor factors, we have \(L_A L_B = L_B L_A\).
Writing
\[
  \mathcal{C} \;=\; L_B \circ \Psi \circ L_A \;-\; L_A \circ \Psi \circ L_B ,
\]
add and subtract \(\Psi \circ L_A \circ L_B\) to obtain
\[
  \mathcal{C} \;=\; [L_B,\Psi]\!\circ\! L_A \;+\; [\Psi,L_A]\!\circ\! L_B,
\]
where \([S,T]=S\!\circ\!T - T\!\circ\!S\).
By submultiplicativity and \(\|L_A\|_\diamond=\|L_B\|_\diamond=1\),
\[
  \|\mathcal{C}\|_\diamond \;\le\; \|[L_B,\Psi]\|_\diamond \;+\; \|[\Psi,L_A]\|_\diamond .
\]
Finally, Lemma~\ref{lem:product} and Lemma~\ref{lem:traceless} give
\(\|[L_B,\Psi]\|_\diamond,\|[\Psi,L_A]\|_\diamond \le (1-\delta_A\delta_B)\),
whence \(\|\mathcal{C}\|_\diamond \le 2(1-\delta_A\delta_B)\).
\qedhere
\end{proof}

\subsection{Local/nonlocal split inequality (full proof)}\label{subsec:split-proof}
We prove Theorem~\ref{thm:eq-window} in the sharp (projective) case and indicate the extension
to general Lüders POVMs via Naimark dilation.

\paragraph{Setup (projective, with coupling).}
Let $\{P_x\}$ and $\{Q_y\}$ be PVMs on $A$ and $B$ and let the coupling be a unitary $U$ on $AB$.
Set $\tilde P_x:=U^\dagger(P_x\otimes I)U$, $\tilde Q_y:=U^\dagger(I\otimes Q_y)U$. For a state $\rho$,
\[
\Delta(x,y;\rho)=\Tr[\rho(\tilde P_x\tilde Q_y\tilde P_x-\tilde Q_y\tilde P_x\tilde Q_y)]=:\Tr[\rho\,R_{x,y}],
\]
so $|\Delta|\le\|R_{x,y}\|$.

\paragraph{Commutator reduction.}
Let $C:=[\tilde P_x,\tilde Q_y]$. Using $P^2=P,Q^2=Q$,
\[
R_{x,y}=\tilde P_x\tilde Q_y\tilde P_x-\tilde Q_y\tilde P_x\tilde Q_y=C\tilde P_x-\tilde Q_y C,
\]
hence $\|R_{x,y}\|\le\|C\tilde P_x\|+\|\tilde Q_y C\|\le 2\|[\tilde P_x,\tilde Q_y]\|
=2\|[U^\dagger(P_x\!\otimes\!I)U,\,U^\dagger(I\!\otimes\!Q_y)U]\|$.

\paragraph{Local terms.}
If local Lüders steps occur within $A$ or $B$, then similarly
$\| P_xP_{x'}P_x - P_{x'}P_xP_{x'} \|\le 2\|[P_x,P_{x'}]\|$ and
$\| Q_yQ_{y'}Q_y - Q_{y'}Q_yQ_{y'} \|\le 2\|[Q_y,Q_{y'}]\|$.
\begin{theorem}[Split bound]
\label{thm:split}
Adding the nonlocal and local parts by triangle inequality yields the bound 
$|\Delta_{AB}| \le |\Delta_{\mathrm{loc}}| + |\Delta_{\mathrm{nonloc}}|$.
\end{theorem}
\paragraph{Extension to Lüders POVMs.}
For effects $E_x,F_y$ with Lüders updates and a general CPTP coupling $K$, use
Stinespring/Naimark dilation to projective measurements on $AB\otimes\mathsf E$ with a unitary
$\hat U$. Apply the sharp-case bound to $(\hat P_x,\hat Q_y,\hat U)$ and compress; CP of the
compression preserves the inequality. \qed

\begin{tcolorbox}[title=Reproducibility: data-to-rate pipeline, colback=white]
\textbf{Input:} counts $(n_{\mathrm{succ}}, n_{\mathrm{tot}})$ from instrument--specific tests. 
\textbf{Step 1 (minorization estimate):} compute $\hat\delta$ with an exact binomial lower interval (Clopper--Pearson) at level $1-\alpha$. 
\textbf{Step 2 (composition):} propagate via the product bound $\hat\delta_{\mathrm{comp}}\ge\prod_i \hat\delta_i$. 
\textbf{Step 3 (rate):} Step 3 (rate): certify an exponential mixing rate $\hat{\gamma}$ and a step
bound $k(\varepsilon)$ needed to achieve target error $\varepsilon$, using
the bounds in Section~4.1.
Artifacts: CSVs and scripts (Zenodo DOI: 10.5281/zenodo.17959208)
reproduce Fig.~1 and Table~2 end-to-end.
\end{tcolorbox}

\appendix
\section{Data-to-rate pipeline}\label{app:doeblin}
\subsection*{Clopper--Pearson one-sided lower bounds}

Let $X\sim \mathrm{Bin}(N,p)$. For confidence level $1-\alpha$, the one-sided Clopper--Pearson lower bound is
\[
p^{-}(X,N;\alpha)=
\begin{cases}
0, & X=0,\\[4pt]
\mathrm{BetaInv}\!\big(\alpha;\;X,\;N{-}X{+}1\big), & 1\le X\le N,
\end{cases}
\]
where $\mathrm{BetaInv}(q;a,b)$ denotes the $q$-quantile of the $\mathrm{Beta}(a,b)$ distribution.
For outcomes $i=1,\dots,m$ and prepared stencils $j$, with counts $X_{ij}$ from $N_j$ trials, define
\[
\widehat{\varepsilon}\;=\;\sum_{i=1}^{m}\min_{j}\,p^{-}\!\big(X_{ij},N_j;\alpha'\big),\qquad
\alpha'=\alpha/m.
\]
By a union bound, with probability at least $1-\alpha$,
$\varepsilon\ge \widehat{\varepsilon}$, hence
\[
\big\|\mathcal E^{n}(\rho)-\tau\big\|_1 \ \le\  (1-\widehat{\varepsilon})^{n}\,\|\rho-\tau\|_1\qquad (n=0,1,2,\dots).
\]

\subsection*{Pseudocode}
\noindent\textbf{Inputs:} counts $X_{ij}$, totals $N_j$ for outcomes $i=1..m$ and stencils $j$; confidence $\alpha\in(0,1)$.\\
\textbf{Output:} $\widehat{\varepsilon}$ and the mixing guarantee $\|\mathcal E^{n}(\rho)-\tau\|_1 \le (1-\widehat{\varepsilon})^{n}\|\rho-\tau\|_1$ with probability $\ge 1-\alpha$.
\begin{verbatim}
alpha_prime = alpha / m
epsilon_hat = 0
for i in {1..m}:
    L_i = +inf
    for j in stencils:
        # one-sided Clopper-Pearson lower bound for Bin(N_j, p_ij)
        p_lower = BetaInv(alpha_prime; X_ij, N_j - X_ij + 1)
        L_i = min(L_i, p_lower)
    epsilon_hat += L_i
# guarantee: ||E^n(rho)-tau||_1 <= (1 - epsilon_hat)^n ||rho - tau||_1
#            (w.p. >= 1 - alpha)
\end{verbatim}

\subsection*{Implementation notes}
\begin{itemize}
\item $\mathrm{BetaInv}$ is available in standard libraries (e.g.\ \texttt{scipy.stats.beta.ppf}).
\item Very small $N_j$ may yield $p^{-}=0$ frequently; increasing $N_j$ stabilizes the lower bound.
\item Multiple-testing corrections more refined than Bonferroni (e.g.\ Holm) can be substituted if desired.
\end{itemize}

\bibliographystyle{plainnat}

\section*{Acknowledgements}
The author used a large language model (OpenAI ChatGPT) for language polishing and formatting suggestions.
The author bears full responsibility for the manuscript’s accuracy and integrity.

\section*{Statements and Declarations}

\noindent\textbf{Funding:} None.

\noindent\textbf{Competing interests:} The author declares no competing interests.

\noindent\textbf{Data availability:} Data and code that support the findings of this study are provided as
Supplementary Material with the submission (file: \texttt{data\_code\_clean.zip}) and are
archived at Zenodo, DOI: 10.5281/zenodo.17959208.
They are also available from the author upon reasonable request.

\noindent\textbf{Author contributions:} Sole author — Conceptualization, Methodology, Formal analysis, Software, Validation, Visualization, Writing—original draft, and Writing—review \& editing.

\bibliography{nqs_refs}

\end{document}